\documentclass[11pt]{article}

\long\def\comment#1{\ifdim\overfullrule>0pt{\sf[{#1}]}\fi}

\usepackage{latexsym}
\usepackage{amsmath, amsfonts, amsthm}
\usepackage{hyperref}
\usepackage{epsfig}
\usepackage{amssymb}

\usepackage{bbm}

\newtheorem{theorem}{Theorem}
\newtheorem{definition}{Definition}
\newtheorem{conjecture}{Conjecture}

\newtheorem{lemma}{Lemma}

\newcommand{\X}{{\cal X}}

\newcommand{\Fp}{{\cal F}_p}
\newcommand{\Se}{S}
\newcommand{\Pe}{P}
\newcommand{\N}{{\cal N}}
\newcommand{\h}{$\frac{1}{2}$}

\newcommand\set[1]{\left\{#1\right\}} 
\newcommand\concept[1]{\textit{#1}} 

\newcommand{{\Ex}}{{{\text{{E}}}}}
\newcommand{{\RR}}{{\sc Rel-Lin-Eq}}
\newcommand{{\RRs}}{{\sc Rel-Lin-Eq }}

\begin{document} 
\bibliographystyle{alpha}
\def\Sketch{\noindent{\bf Sketch:\ \ }}
\def\eps{\epsilon}

\title{Rounding semidefinite programs for large-domain problems \\via Brownian motion}

\author{Kevin L. Chang\thanks{Research done at Max-Planck-Institut f\"ur Informatik, 
  Saarbr\"ucken, Germany.  {\tt kchang@mpi-inf.mpg.de}.}
\and Alantha Newman\thanks{Research done at
  Max-Planck-Institut f\"ur Informatik, Saarbr\"ucken, Germany.
CNRS-Universit\'e Grenoble Alpes. {\tt
    alantha.newman@grenoble-inp.fr}.}}


\maketitle

\begin{abstract}
We present a new simple method for rounding a semidefinite programming
relaxation of a constraint satisfaction problem.  We apply it to the
problem of approximate angular synchronization studied in
\cite{makarychev2011complex}.  Specifically, we are given directed
distances on a circle (i.e., directed angles) between pairs of
elements and our goal is to assign the elements to positions on a
circle so as to preserve these distances as much as possible.
The feasibility of our rounding scheme is based on properties of the
well-known stochastic process called Brownian motion.  Based on
computational and other evidence, we conjecture that this rounding
scheme yields an approximation guarantee that is very close to the
best-possible guarantee (assuming the Unique-Games Conjecture).
\end{abstract}

\section{Introduction}

We present an alternate approach to the {\sc Relaxed Linear Equations
  $\bmod ~p$} ({\sc Rel-Lin-Eq}) problem studied in
\cite{makarychev2011complex}.  Our new approach is also based on
rounding a semidefinite programming relaxation but uses a 
different rounding technique.  Based on computational evidence and
other justification, we believe this approach has essentially the same
approximation guarantee of $.854$ for {{\sc Rel-Lin-Eq}} as proven for
a different algorithm presented in \cite{makarychev2011complex}.

We are given a set $E$ of equations in the form of 
$x_j -x_i \equiv d_{ij} (\bmod ~p)$.  
Let $\X = \{x_i\}$ be the set of elements and let
$n = |\X|$.  We assign each element in $\X$ a (integral) value from the
set $[0,p)$.  For a fixed assignment, an equation has value $x_j - x_i
  \equiv d_{ij} \pm y_{ij} (\bmod ~p)$, for $y_{ij} \leq p/2$.  Since
  $y_{ij}$ can have a nonnegative value of at most $p/2$, we divide
  $y_{ij}$ by $p/2$ in order to obtain a normalized value between 0
  and 1.  Our goal is to find an assignment that maximizes the sum
  $\sum_{ij \in E} (1 - 2 y_{ij}/s)$.  More precisely, we formulate
  our objective function as follows.
\begin{eqnarray}
\max \sum_{ij \in {E}} \left(1-\frac{2\cdot \min\{
(x_j-x_i-d_{ij}) \bmod p,~ -(x_j -x_i - d_{ij}) \bmod p \}}{p} \right). \nonumber
\end{eqnarray}

This problem generalizes the {\sc Max-Cut} problem: For any edge $ij$
in a graph, we write the constraint $x_j -x_i \equiv p/2 (\bmod ~p)$.
Then an $\alpha$-approximation to {\sc Rel-Lin-Eq} yields an
$\alpha$-approximation for {\sc Max-Cut}.  It can be viewed as an
approximate version of the {\em angular synchronization} problem
studied by Singer~\cite{singer2011angular}.  Originally, our
motivation was to develop rounding methods for constraint satisfaction
problems whose solutions to assignment-constraint based SDPs can have
the property that no pair of assignment vectors have a high dot
product despite the solution having a high objective value.  (e.g., A
problem with this property is {\sc Maximum Acyclic Subgraph}, which
has since been proven to be Unique-Games hard to approximate to within
a factor greater than $\frac{1}{2}$~\cite{guruswami2011beating}.  On
the other hand, {\sc Unique Games} does not have this property.)

Our rounding scheme is based on properties of the well-known
stochastic process called Brownian motion.  Theoretically, this
procedure could be applied to other problems that can be modeled using
the standard assignment-constraint based SDP framework (i.e., SDP
formulation $(P^+)$ in Section \ref{SDP}).  However, it does seem
tailor-made for our particular objective function.  It is also
reminiscent of ``sticky'' random walks used in constructive approaches
to discrepancy minimization~\cite{bansal2010constructive}, although these results
focus on assigning each element a binary value and one of our main
motivations was to study how to approximate large-domain problems.

\subsection{Organization}

After presenting our quadratic formulations and relaxations in Section
\ref{sec:quad-prog}, we present our rounding procedure in Section
\ref{sec:round}.  In Section \ref{sec:brownian}, we discuss Brownian
motion and how it relates to our rounding procedure.  Then we prove
that our rounding procedure is feasible most of the time; precisely,
at least $.96$ of the time, it assigns a (non-random) position to a
variable.  First, we prove this for a continuous process (Section
\ref{sec:sign_change}) and then for a discrete process (Section
\ref{sec:discrete}).  Finally, in Section \ref{sec:correlated}, we
state a conjecture regarding the correlation of two random walks,
which is supported by extensive computational experiments.  A positive
resolution to this conjecture would be one way to prove that this
rounding procedure has a guarantee close to the best-known guarantee
of $.854$~\cite{makarychev2011complex} (and close to the best-possible
guarantee of $.878$ under the Unique-Games Conjecture).

\section{Quadratic Programs}\label{sec:quad-prog}

For each variable $x_i$, we have a set of $p$ unit vectors, $v_i^0, v_i^1,
v_i^2, \dots, v_i^{p-1}$, for a total of $pn$ vectors.  
For $b > a$, let $d(a,b) = \min\{b-a, p - (b-a)\}$.  Note that $d(a,b) =
d(b,a)$.  Let $\Fp$ denote a particular (fixed) set of $p$ vectors with
the property that for $v_a, v_b \in \Fp$, $v_a \cdot v_b =
1-\frac{4 d(a,b)}{p}$.  For example, if $p = 8$, then the
set $\Fp$ can be the following eight vectors.

\begin{center}
\begin{tabular}{cccccccc}
$v_i^0$ & $v_i^1$ & $v_i^2$ & $v_i^3$ & $v_i^4$ & $v_i^5$ & $v_i^6$ &
  $v_i^7$ \vspace{2mm}
\\
\h & -\h & -\h & -\h & -\h & \h & \h & \h \vspace{1.5mm}\\
\h & \h & -\h & -\h & -\h & -\h & \h & \h \vspace{1.5mm}\\
\h & \h & \h & -\h & -\h & -\h & -\h & \h \vspace{1.5mm}\\
\h & \h & \h & \h & -\h & -\h & -\h & -\h 
\end{tabular}
\end{center}
This formula for creating such a set $\Fp$ of vectors can be
generalized for any even value of $p$ (where the absolute value of
each coordinate of each vector is $\frac{\sqrt{2}}{\sqrt{p}}$).  We
obtain the following quadratic program for \RR.  Let $\Pe$ denote the
set of integers in $[0,p)$.

\vspace{4mm}
\fbox{\parbox{14cm}{
{\bf A Quadratic Program ($\mathit{Q}$):}
\begin{align}
 \max & \sum_{ij \in E} \frac{1 + v_i^0 \cdot v_j^{d_{ij}}}{2} \nonumber\\
v_i^a \cdot v_i^{b} & =   1- \frac{4\cdot d(a,b)}{p}, & \forall x_i \in
\X, ~a,b \in \Pe, \label{constell}\\
v_i^a \cdot v_j^b & =  v_i^{k+a}\cdot v_j^{k+b},  & \forall x_i,x_j
\in \X,~a,b,k \in \Pe, \label{equal-vec}\\
v_i^a \cdot v_i^a & = 1, &  \forall x_i \in \X, a
\in \Pe,
\label{unit-vec}\\
v_i^k & \in \Fp, & \forall x_i \in \X, k \in
\Pe. \label{strict-enforce}
\end{align}}}
\vspace{4mm}

For each variable $x_i \in \X$, the corresponding set of $p$ vectors
has the same configuration up to rotation, reflection and translation.
This is enforced by Constraints \eqref{constell} and \eqref{unit-vec}.
In an integral solution, the set of $p$ vectors corresponding to
variable $x_i$ is identical to the set of $p$ vectors corresponding to
the variable $x_j$, for all variables $x_i, x_j \in \X$.  This follows
from the fact that all vectors belong to $\Fp$.  The only difference is
that vectors in the two sets may have different labels (i.e., one set
of vectors can be viewed as a rotation of the other set).  Then the
relative values or positions of two variables only depends on the
rotations of the labels.  In other words, if for variables $x_i$ and
$x_j$, the same vectors have the same labels, then the two variables
will be assigned to the same position.  If $v_i^k = v_j^{k+3}$ and
$x_i$ is in position $k$, $x_j$ should be in position $k+3$.  Given an
integral solution, we can determine the position of each variable by
picking a vector, and assigning each variable to the label to which
that vector corresponds for that variable.

\subsection{Semidefinite Relaxations}\label{SDP}

To obtain a semidefinite relaxation of ($Q$), we remove Constraint
\eqref{strict-enforce} and require only that each $v_i^k \in
      {\mathbbm{R}}^{pn}$.  Note that even in the semidefinite
      relaxation, the set of $p$ vectors corresponding to a particular
      variable $x_i$ have the same configuration for each variable up
      to rotation, reflection and translation.  We refer to the set of
      vectors corresponding to a variable $x_i$ as a {\em
        constellation} $C_i$.  We show that certain properties hold
      for each constellation.

\vspace{4mm}
\fbox{\parbox{14cm}{
{\bf A Semidefinite Program ($\mathit{P}$):}
\begin{align}
\max & \sum_{ij \in E} \frac{1 + v_i^0 \cdot v_j^{d_{ij}}}{2} & \nonumber\\
v_i^a \cdot v_i^{b} & =   1- \frac{4\cdot d(a,b)}{p}, &  \forall
x_i \in \X, ~a,b \in \Pe, \label{constell_sdp}\\
v_i^a \cdot v_j^b & =  v_i^{k+a}\cdot v_j^{k+b}, &  \forall x_i,x_j
\in \X,~a,b,k \in \Pe, \label{equal-vec_sdp}\\
v_i^a \cdot v_i^a & =  1, &  \forall x_i \in \X, a
\in \Pe,
\label{unit-vec_sdp}\\
v_i^k & \in {\mathbbm{R}}^{pn}, & \forall x_i \in \X, k \in \Pe.\label{strict-enforce_sdp}
\end{align}}}
\vspace{4mm}

Let ${\bf v_p}$ be a vector in ${\mathbbm{R}}^{\frac{p}{2}}$ in which
each entry is $\frac{\sqrt{2}}{\sqrt{p}}$.  Let ${\bf e}_i \in
\mathbbm{R}^{\frac{p}{2}}$ be the indicator vector which has a $1$ in
the $i^{th}$ position and 0 elsewhere.  For $k$ such that $0 \leq k
\leq p/2$, we define $w_k$ as follows:
$$w_k = \sum_{i=0}^{k} \frac{\sqrt{2}}{\sqrt{p}}{\bf e}_i.$$

\begin{definition}
Let the constellation $C_0$ be the set of $s$ unit vectors $\{v_0^0, v_0^1,
v_0^2, \dots v_0^{p-1}\}$ defined as follows.  For $k$ such that $0
\leq k \leq p/2$, define $v_0^k = -2 \cdot w_k + {\bf v_p}$.  For $k$
such that $p/2 < k < p$, let $v_0^k = - v_0^{k-p/2}$.
\end{definition}

\begin{lemma}\label{same_constell}
For any $x_i \in \X$, the constellation $C_i = \{v_i^0, v_i^1,
\dots, v_i^{p-1}\}$ is equivalent to
the constellation $C_0$ up to rotation, reflection and translation. 
\end{lemma}

\begin{proof} Let $v_{ik} = \frac{(v_i^k - v_i^{k-1})}{2}$.  Without loss of
generality, let us assume that $v_{ik} = \frac{\sqrt{2}}{\sqrt{p}}
{\bf e}_k$ for $1 \leq k \leq p/2$.  We can assume this since for all
$k \in \Pe$, we have (i) $||v_{ik}|| = \frac{\sqrt{2}}{\sqrt{p}}$
(Lemma \ref{lemm_length}) and (ii) $v_{ij} \cdot v_{ik} = 0$ for all
$j,k \in \Pe$ such that $j \neq k$ (Lemma \ref{lemm_overlap}).

Note that $v_i^k = v_i^{k-1} + 2 \cdot v_{ik}$.  
This implies that $v_i^{p/2} = - \sum_{k=1}^{p/2} v_{ik}$ and that  
$v_i^0 = \sum_{k=1}^{p/2} v_{ik}$.  Thus, there is some rotation of
the vectors in $C_i$ such that the resulting set of vectors is
equivalent to $C_0$.\end{proof}

\begin{lemma}\label{lemm_length}
For all $x_i \in \X$ and $k \in \Pe$, $||\frac{(v_i^k-v_i^{k-1})}{2}||=
\frac{\sqrt{2}}{\sqrt{p}}$.
\end{lemma}
\begin{proof}
\begin{eqnarray}
\left(\frac{(v_i^k-v_i^{k-1})}{2}\right)\cdot \left(\frac{(v_i^k-v_i^{k-1})}{2}\right)
& = &  \frac{1}{4} (v_i^k-v_i^{k-1})(v_i^k-v_i^{k-1}) \nonumber \\ &
= & \frac{1}{4} (v_i^k \cdot v_i^k +
v_i^{k-1}\cdot v_i^{k-1} - 2v_i^k \cdot v_i^{k-1}) \nonumber\\
& = & \frac{1}{4}(2 - 2(1- \frac{4}{p})) \nonumber \\
& = & 2/p. \nonumber 
\end{eqnarray}
\end{proof}

\begin{lemma}\label{lemm_overlap}
  For $x_i \in \X$ and for $a,b,c,d \in [0,p/2]$, the vectors
  $(v_i^a-v_i^b)\cdot(v_i^c - v_i^d) = 0$ if $[a,b]$ and $[c,d]$
  are non-overlapping intervals.
\end{lemma}
\begin{proof}
\begin{eqnarray}
(v_i^a-v_i^b)\cdot(v_i^c - v_i^d) & = & v_i^a \cdot v_i^c - v_i^a
\cdot v_i^d - v_i^b \cdot v_i^c + v_i^b\cdot v_i^d \nonumber \\
& = &  -\frac{4 d(a,c)}{p} + \frac{4 d(a,d)}{p} + \frac{4 d(b,c)}{p}
-\frac{4 d(b,d)}{p}.\nonumber
\end{eqnarray}
This equals 0 if the intervals $[a,b]$ and $[c,d]$ are non-overlapping, since then $d(a,c) +
d(b,d) = d(a,d) + d(b,c)$.\end{proof}

Another way to write a semidefinite program is to use a standard formulation based on assignment
constraints (e.g., see Quadratic Program $(Q_2)$ in
\cite{makarychev2011complex}).
  
\vspace{4mm}
\fbox{\parbox{14cm}{
{\bf A Semidefinite Program ($\mathit{P^+}$):}
\begin{align}
\max \sum_{ij \in E} \sum_{k \in \Pe} & \left(p- 2 d(k,d_{ij})\right) u_{i0} \cdot u_{j k} \nonumber\\
u_{ih}\cdot u_{jk} & \geq 0, & x_i,x_j \in \X,
h,k \in \Pe,\label{non-neg-sdp}\\
u_{ih} \cdot u_{ik} & = 0, & x_i,x_j \in \X,
h,k \in \Pe,\\
u_{ih} \cdot u_{jk} & = u_{ih+a}\cdot u_{jk+a}, & x_i,x_j \in \X,
h,k,a \in \Pe,\\
u_{ih} \cdot u_{ih} & = \frac{1}{p}, & \forall x_i \in \X,\\
|\sum_{h \in \Pe} u_{ih} - \sum_{k \in \Pe}u_{jk}|^2 & = 0, & \forall x_i, x_j \in \X,\\
u_{ih} & \in \mathbbm{R}^{pn}, & \forall x_i \in \X, h \in \Pe.
\label{unit-vec_sdp2}
\end{align}}}
\vspace{4mm}

Given a solution for $(P^+)$, we can
construct a solution for $(P)$ as follows.
\begin{eqnarray}
v_i^k & = & \sum_{h=k}^{k+p/2-1} u_{ih} - \sum_{h=k+p/2}^{k+p-1} u_{ih}.\label{vec-relation}
\end{eqnarray}
It is not difficult to see that the transformation in
\eqref{vec-relation} preserves the objective value.
In our computational experiments, we used solutions for $(P^+)$, which
are more constrained than solutions for $(P)$ (e.g., Constraint
\eqref{non-neg-sdp} is not implied by the constraints in $(P)$).
However, we feel it is somewhat clearer to present the rounding
algorithm in the next section based on a solution for $(P)$.

\subsection{Relaxation on an Arbitrarily Large Domain}

Note that we can replace $p$ with an arbitrarily large constant $s$.
Suppose $s$ is a multiple of $p$ (i.e., $s = \ell p$).  Then we can scale
each constraint so that $x_j -x_i \equiv d_{ij} (\bmod ~p)$ becomes
$x_j -x_i \equiv \ell d_{ij} (\bmod ~s)$.  The optimal objective value of
the original and the scaled problem are the same.  Moreover, given a
solution for $(P)$ on the domain of size $p$, we can create a solution
for the scaled problem on the domain of size $s = \ell p$ with the same
objective value without resolving the relaxation $(P)$.
Thus, we can assume that $s$ is an extremely large constant.  We assume this
since our rounding algorithms work best on a large domain. 

\section{Rounding the Relaxation}\label{sec:round}

Our algorithm for \RRs is based on rounding a solution for the semidefinite
relaxation ($P$) presented in Section \ref{SDP}.  The first issue is,
how do we use the constellation of vectors $C_i$ to determine the
position or value of variable $x_i$?  We will consider the following
random process with $s$ steps.  Let $r \in {\mathbbm{R}}^{sn}$ be a
random vector in which each coordinate is chosen according to the
normal distribution $\N(0,1)$.  We can view the $s$ values $r\cdot
v_i^0, r \cdot v_i^1, \dots, r \cdot v_i^{s-1}$ as a discrete random process in which the expected correlation
of $r \cdot v_i^a$ and $r \cdot v_i^{b}$ is given by the dot product
$v_i^a \cdot v_i^b$.

Let us view these $s$ values as a discrete random process on the
interval $[0,s]$.  For a subinterval $[t, t']$, we say time step $q$
is in the interval $[t,t']$ if $d(s\cdot t/2, s\cdot q/2) \leq
d(s\cdot t/2, s\cdot t'/2)$ and $d(s\cdot t'/2, s \cdot q/2) \leq
d(s\cdot t/2, s\cdot t'/2)$.

Given such a random process, we say that there is an {\em extreme sign
  change} with threshold $\alpha$ between times $t$ and $t'$ if
$v^i_t\cdot r \leq -\alpha$, $v^i_{t'}\cdot r \geq \alpha$ and
$v^i_{q} \cdot r < \alpha$ for all $q \in [t,t']$.  Our algorithm is
based on the observation that in this random process, it is very
likely that there is exactly one extreme sign change for the threshold
$\alpha = 1$ (i.e., there do not exist two disjoint intervals that
both contain extreme sign changes).  This is stated in Theorem
\ref{sign_change_thm}.  If this random process has exactly one extreme
sign change, then we say that the process {\em first reaches a
  threshold $\alpha$ at time $t$} if there is an interval $[t',t]$
such that $v^i_t \cdot r \geq \alpha$, $v^i_{t'} \cdot r \leq
-\alpha$, and $v^i_{q} \cdot r < \alpha$ for all $q \in[t',t]$.  Note
that these intervals are taken modulo $s$ (i.e., they are intervals on
a circle).

\begin{definition}
An {\em extreme sign change} with threshold $\alpha$ in the sequence
$\{v_i^0 \cdot r, v_i^1 \cdot r, \dots, v_i^{s-1} \cdot r\}$ occurs
when $v^i_{t_1} \cdot r \leq -\alpha$ and $v^i_{t_2} \cdot r \geq \alpha$ and for no
value of $t: t_1 < t < t_2$ is $v^i_t \cdot r \geq \alpha$.
\end{definition}

\begin{theorem}\label{sign_change_thm}
  If $s$ is a sufficiently large constant, then with probability at
  least $.96$, the random process $\{v_i^k \cdot r\}$ with $s$ steps
  has exactly one extreme sign change with threshold 1.
\end{theorem}

\begin{figure}[t]
\begin{center}
  \epsfig{file=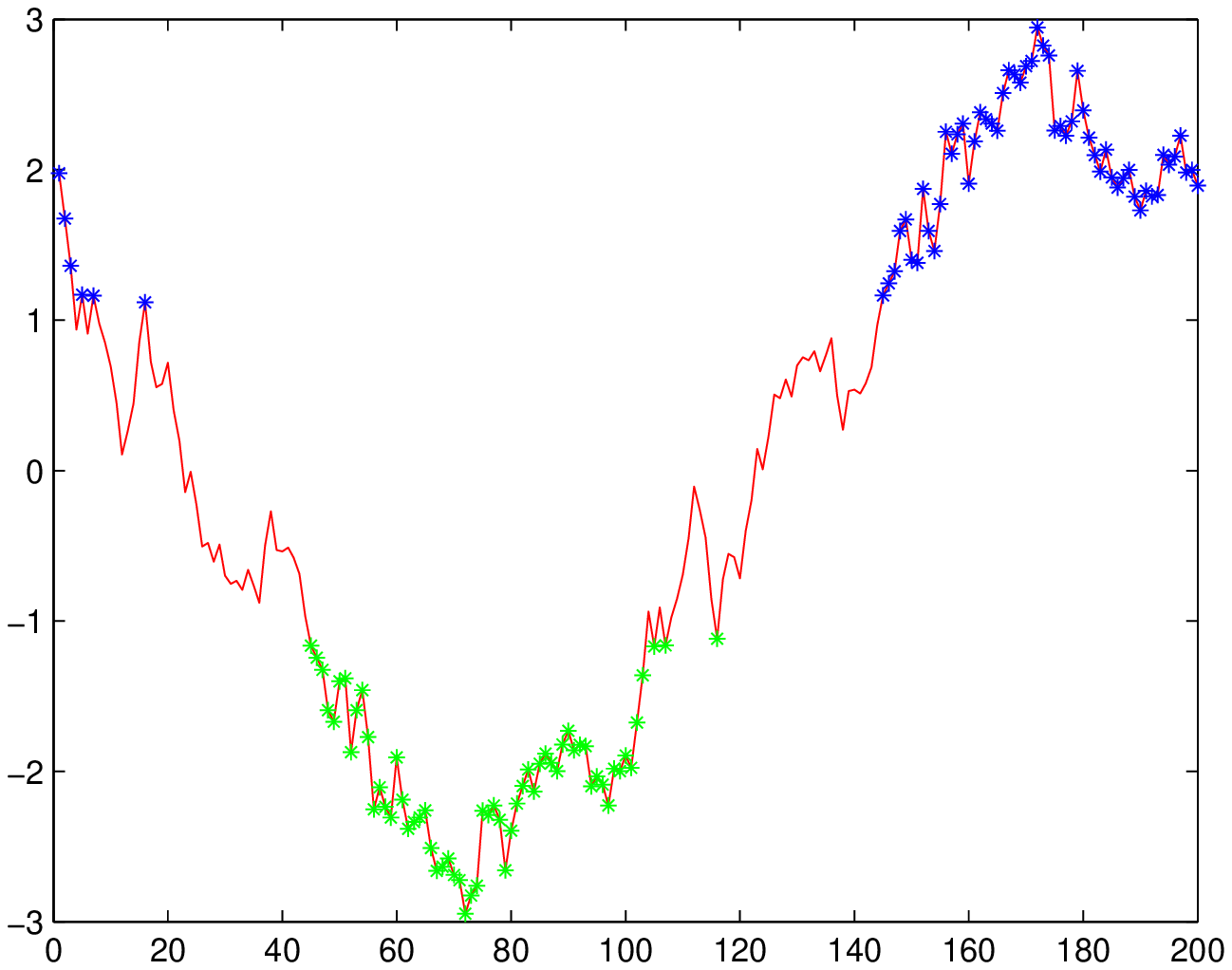, width=5.4cm} \epsfig{file=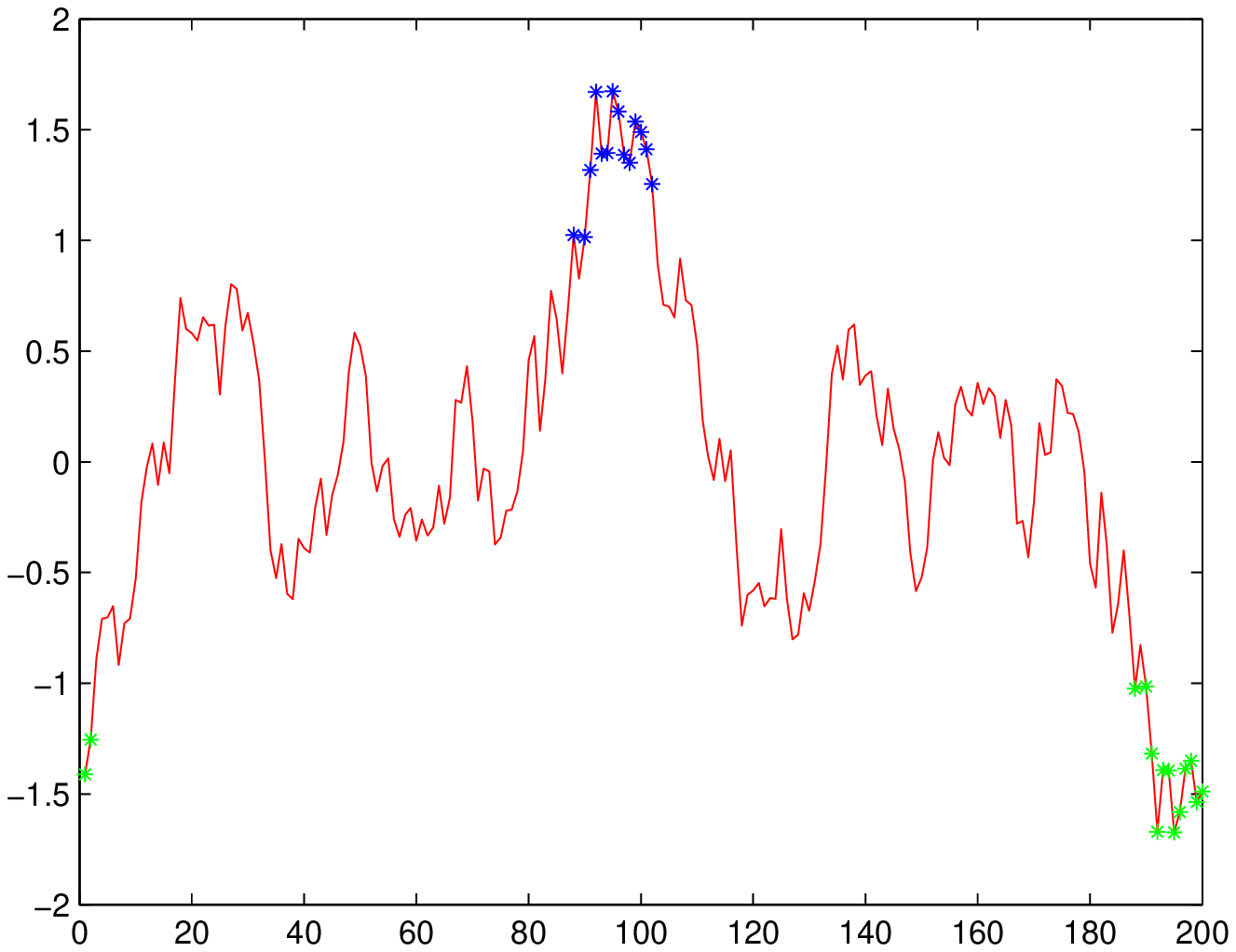,
    width=5.4cm} \epsfig{file=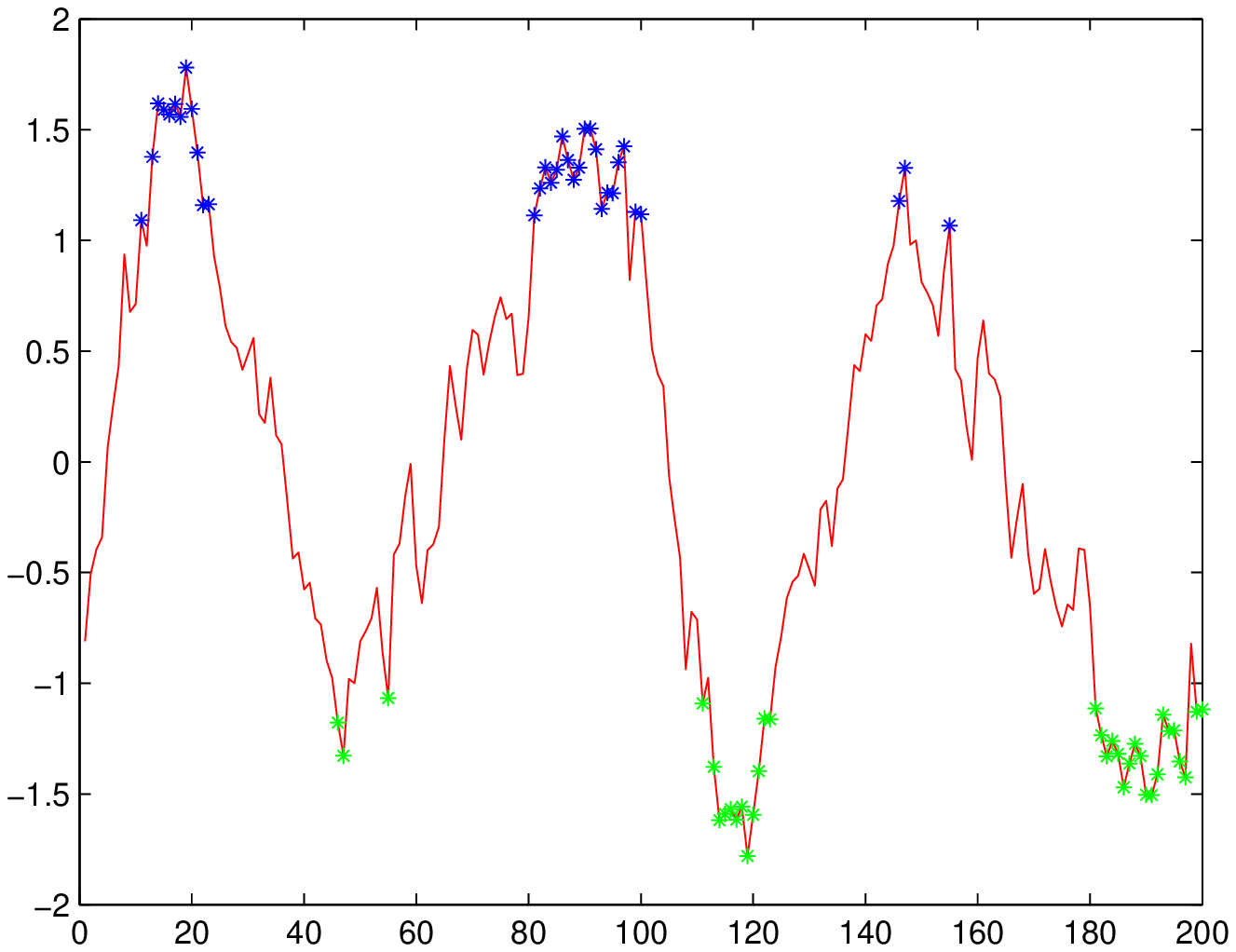, width=5.4cm}
  {\caption{The first two figures depict instances with one extreme
      sign change for threshold 1.  The third figure shows an instance
      with three extreme sign changes for threshold 1.  The blue dots
      represent the values greater than 1 and the green dots represent
      the values less than -1.
\label{signs}}}
\end{center}
\end{figure}

\subsection{Rounding Algorithm}

Given Theorem \ref{sign_change_thm}, we present the following
rounding algorithm.

\vspace{4mm}

\noindent
\fbox{\parbox{16.5cm}{
\begin{itemize}

\item[(i)] Solve the semidefinite relaxation ($P$).

\item[(ii)] Choose a random vector $r \in {\mathbbm{R}}^{sn}$ where
  each coordinate $r_i \sim {\cal N}(0,1)$ for $i \in \{1,\dots, sn\}$.  

\item[(iii)] {For each variable $x_i \in X$, consider the sequence
    $\{v_i^k \cdot r\}$ for all $k \in \Se$.

\begin{itemize}

\item[(a)] {If there is one extreme sign change:
\begin{itemize}
\item[] Place $x_i$ at the $k$ where $\{v_i^k \cdot r\}$ first reaches
  threshold 1.
\end{itemize}}

\item[(b)] {If there are no extreme sign changes:

\begin{itemize}
\item[] Assign
  $x_i$ a random position in $[0,s-1]$.
\end{itemize}
}

\end{itemize}}
\end{itemize}}}

\vspace{3mm}

For each $x_i \in X$, we associate the random walk $w^i$.  Each walk
$w^i$ has the same expected behaviour.  This follows from the fact
that for each $i$, there is a rotation matrix such that the set of
vectors $\{v_i^k\}$ is equal to a canonical set of vectors, as stated
in Lemma \ref{same_constell} (i.e., for a fixed vertex $i$, the
pairwise (setwise) relationships of the vectors is exactly prescribed
by constraints (2) and (4) of the SDP).  We prove Theorem
\ref{sign_change_thm} in Section \ref{sec:sign_change}.  First, we
briefly discuss {\em Brownian motion}, which we will use in the proof
of Theorem \ref{sign_change_thm}.

To measure the quality of the solution produced by this rounding
algorithm, one must analyze the correlation of two random walks.  In Section
\ref{sec:correlated}, we state a conjecture regarding this correlation,
which is supported by extensive computational experiments.

\section{Brownian Motion}\label{sec:brownian}

In order to analyze the randomized rounding scheme presented above, we
will interpret the sequence of vectors corresponding to any fixed
variable, $v_1\cdot r, v_2\cdot r,\ldots, v_s\cdot r$ as a random
walk.  We will show that this random walk is a discrete sampling of a
fundamental continuous stochastic process called \concept{Brownian
  motion}.  We will then use properties of Brownian motion to prove
properties of our discrete walk.  For background in Brownian motion,
we refer the reader to the textbook \cite{Brownian}.

A stochastic process $W_t$ for $0\le t<\infty$
is a Brownian motion if it satisfies the following
properties:  
\begin{enumerate}
\item For times $t_1<t_2$, $W_{t_2}-W_{t_1}\sim \N(0, t_2-t_1)$.
\item For all choices of times $0\le t_0<t_1<\ldots< t_n<\infty$,
$W_{t_{i+1}}-W_{t_i}$ is independent of $W_{t_{j+1}}-W_{t_j}$
for all choices of $i, j, n$.  
\item $W_0 = 0$. 
\item $W_t$ has continuous sample paths with probability 1.  
\end{enumerate}

Our proofs will rely on two basic properties of Brownian motion:
the \concept{distributions
of hitting times} and the \concept{Reflection Principle}.

The \concept{first hitting time} for level $b$, denoted by $\tau_b$, 
is defined to be the first time at which $W_t$ takes
the value $b$:  $\tau_b = \inf\set{t:W_t = b}$.  
$\tau_b$ is a random variable whose distribution has
the density function:
\begin{equation}
\Pr[\tau_b \in dt] = \frac{b}{\sqrt{2\pi}t^{3/2}}
\exp\left(-\frac{b^2}{2t}\right)dt.
\end{equation}

Roughly stated, the Reflection Principle is the intuitive 
property that once a Brownian motion hits a level $b$, it
is equally likely to be above and below the level $b$ in the future.
More precisely, it states that if $W_t$ is a Brownian motion
and $\tau_b$ its hitting time to $b$, then
the process $W'_t$ defined by 
\begin{equation*}
W'_t=\left\{\begin{array}{ll}
  W_t & \mbox{for $0\le t\le \tau_b$}\\
  2b-W_t & \mbox{for $t\ge \tau_b$}
  \end{array}
\right.
\end{equation*}
(which is the formula for $W_t$ ``reflected''
about the horizontal line $y=b$ after it hits $b$) is also a Brownian motion.  
We will use the reflection principle many times in our calculations.

\subsection{Mapping Our Process to Brownian Motion}

Given a constellation of vectors $C_i$, we can assume by Lemma
\ref{same_constell} that the vectors have the following
configuration.  Each vector in this configuration has
dimension $s/2$.   Note that in order to make each vector a unit
vector, we can multiply each entry by $\frac{\sqrt{2}}{\sqrt{s}}$.
\begin{center}
\begin{tabular}{cccccc}
$v_i^{0'}$ & $v_i^{1'}$ & $v_i^{2'}$ & $v_i^{3'}$ & $\dots$ & $v_i^{{s-1}'}$ \\
1 & -1 & -1 & -1 & \dots & 1 \\
1 & 1 & -1 & -1 & \dots & 1 \\
1 & 1 & 1 & -1 & \dots & 1 \\
1 & 1 & 1 & 1 & \dots & 1 \\
\dots & \dots & \dots & \dots & \dots & \dots
\end{tabular}
\end{center}
Suppose that $r \in \N(0,1)^{s/2}$ is a vector such that $r = (r_1,
r_2, \dots r_{s/2})$.  Let 
$a' = \sum_{i=1}^{s/2} r_i = v_i^{0'} \cdot r$.
Define the 
process $w_k^{i'}$ as $w_k^{i'} = \sum_{i=0}^{k} r_k$.  Thus, we have that: 
\begin{eqnarray*}
a' - 2 w_k^{i'}  = v^{k'} \cdot r, \quad w_k^{i'} = \frac{a' - v^{k'} \cdot r}{2}.
\end{eqnarray*} 
Let the vector $v^k$ represent the vector $v^{k'}$ with each entry
multiplied by $\frac{\sqrt{2}}{\sqrt{s}}$, so that $\{v^k\}$ is a set
of unit vectors.  Let $a = \sum_{i=1}^{s/2} r_i \cdot
\frac{\sqrt{2}}{\sqrt{s}} = v^{0} \cdot r = (v^{0'} \cdot r) \cdot
\frac{\sqrt{2}}{\sqrt{s}}$, and define the process $w_k^i$ as
$w_k^i = \sum_{i=0}^{k} r_k \cdot \frac{\sqrt{2}}{\sqrt{s}}$.  Thus, we have that: 
\begin{eqnarray*}
a - 2 w_k^i  = v^{k} \cdot r, \quad w_k^i = \frac{a - v^{k} \cdot r}{2}.
\end{eqnarray*} 
Note that when $v^k \cdot r = 1$ or
$-1$, it is the case that $w_k^i = \frac{a-1}{2}$ and $\frac{a+1}{2}$,
respectively.  If we define $W_t$ to be a Brownian motion on the
continuous interval $[0,1]$, then 
$w_t$ is a discretization of this continuous process.  In the
remainder of the paper, when we refer to a particular process $w_k^i$,
we will drop the superscript $i$ when it is clear from the context, or
when we are just referring to a single process.  

\section{Probability of Exactly One Extreme Sign Change}\label{sec:sign_change}

In this section, we
we prove Theorem \ref{sign_change_thm}.  
We adopt standard statistical notation and denote
\concept{density function} of a continuous random
variable $X$ by $\Pr[X\in dx]$.  For example, if $X\sim \N(0, 1)$,
then $\Pr[X\in dx] = \frac{1}{\sqrt{2\pi}}e^{-x^2/2}dx$.
(Heuristically, $dx$ denotes a very small region around the value
$x$.)  We will furthermore write $\Pr[X\in dx(1)]$ to denote the
density function of $X$ when $x$ takes the specific value of 1.

We compute the probability of an event $A$ conditioned on a random
variable $X$ taking value $X=x$ by applying the formula:
\begin{equation*}
\Pr[A|X = x] = \frac{\Pr[A \mbox{ and } X\in dx]}{\Pr[X\in dx]}.
\end{equation*}

\subsection{Probability of at Least One Sign
  Change}\label{sec:one-change} Suppose Brownian motion $W_t$ begins
at $W_0 = 0$ and at $t = 1$ satisfies $W_1 = a$.  Let $H^+$ be the
minimum time that $W$ reaches the threshold $a/2+1/2$, and $H^-$ be
the minimum time that $W$ reaches $a/2-1/2$.  Note that
$H^+=\tau_{W_1/2+1/2}$ and that $H^-=\tau_{W_1/2-1/2}$.  These hitting
times depend on the value of $W_t$ at time $1$, and are therefore are
not stopping times.  In order to calculate their distributions, we
will first fix $a$ and then calculate the distributions of $H^+$ and
$H^-$,
\textit{conditioned} on the value of $W_1=a$.  

For our proof, we will find it helpful to generalize our definition of
$H^-$ and $H^+$ as follows.  Suppose a Brownian motion satisfies $W_1
= a$.  Define $H^+_{i}$ to be the first time that $W_t$ finishes
hitting the following sequence of barriers: $a/2+1/2$, then $a/2-1/2$,
then $a/2+1/2$ and so forth until it has crossed $i$ barriers,
alternating between the upper and lower barriers.  As an example,
$H^+_3$ would be the first time that the path hits
$a/2+1/2$ \textit{after} having first hit $a/2+1/2$, and then
$a/2-1/2$.  Similarly, define $H^-_i$ to be the first time that $W_t$
finishes hitting the sequence of barriers: below $a/2-1/2$, above
$a/2+1/2$, and so forth until it has crossed $i$ barriers.

\begin{lemma}\label{lem:one-change}
$\Pr[H^+\le 1 \mbox{ or } H^-\le 1]\ge .985612$.
\end{lemma}
Our approach will be to consider the decomposition
of the total probability into probabilities conditioned
on $W_1=a$:
\begin{equation*}
\Pr[H^+ \le 1 \mbox{ or } H^-\le 1]=
\int_{-\infty}^{\infty}\Pr[H^+\le 1 \mbox{ or } H^-\le 1~|~W_1 = a]\phi(a)da.
\end{equation*}
and calculate the integral on the right-hand-side.  

We break the domain of the integral into three parts.

\subsubsection{Case (i): $a\ge1$}
In this case, the upper barrier is $a/2+1/2 \le a$.  
The condition $W_1=a$ implies that 
$W_t$ crosses this barrier with probability 1 (i.e. $H^+\le 1$).  Therefore,
\begin{equation*}
\int_{1}^{\infty}\Pr[H^+\le 1 ~|~W_1 = a]\phi(a)da = 
\int_{1}^{\infty}\phi(a)da \ge .158655.
\end{equation*}

\subsubsection{Case (ii): $a\le -1$}
Analogous to Case (i).  

\subsubsection{Case (iii):  $-1< a< 1$}
From an application of the 
inclusion-exclusion principle, note that:
\begin{eqnarray*}
\Pr\left[H^+\le 1 \mbox{ or } H^{-}\le 1~|~W_1 = a\right] &=& 
\Pr[H^+\le 1~|~W_1 = a]+\Pr[H^-\le 1~|~W_1 = a]\\ & &
-\Pr[H^+\le 1 \mbox{ and } H^-\le 1~|~W_1 = a].
\end{eqnarray*}

\subsubsection{An example of applying the reflection
principle} We now sketch the reasoning behind a standard calculation
involving the reflection principle of Brownian motion and apply it to
calculating $\Pr[H^+\le 1~|~W_1 = a]$ for the case $-1<a<1$.  We will
use this sort of calculation many times in our proofs.

Fix a value of $a$.  Suppose a Brownian motion $W_t$ (but not
restricted to satisfy $W_1 = a$) hits the value $b = a/2+1/2$ at time
$\tau_b$ (i.e., $W_{\tau_b} = b$).  Define $W'_t$ to be the process
$W'_t = W_t$ for $0\le t\le \tau_b$ and $W'_t= 2b-W_t$ for
$t\ge \tau_b$.  By the reflection principle, the random process $W'_t$
is also a Brownian motion.

Now consider the subset of Brownian motions $\set{W_t|H^+\le 1, W_1 =
a}$ (i.e., they satisfy $\tau_{a/1+1/2}\le 1, W_1 = a$).  Note that
these processes correspond exactly to reflected processes $W'$ that
satisfy $W'_1 = 2b-W_1 = 1$.  Thus, the elements in the set
$\set{W_t|H^+\le 1, W_1 = a}$ correspond exactly to elements in the
set $\set{W'_t|W'_1 = 1}$.
Then  
\begin{eqnarray*}
\Pr[H^+\le 1~|~W_1 = a] = 
\frac{\Pr[W'_1 \in dx(1)]}{\Pr[W_1 \in da]}
=\frac{\phi(1)}{\phi(a)}.
\end{eqnarray*}
For a more rigorous justification of these calculations, see \cite{jChang}.

Similarly, one can show that:
$\Pr[H^-\le 1~|~W_1 = a] = \phi(1)/\phi(a)$.  
Therefore
\begin{equation*}
\int_{-1}^{1}\Pr[H^+\le 1~|~W_1 = a]\phi(a) da = 
\int_{-1}^{1}\Pr[H^-\le 1~|~W_1 = a]\phi(a) da = 
\int_{-1}^1\phi(1)da \ge .483941.
\end{equation*}

\subsubsection{An example applying the reflection principle twice}\label{sec:twice-reflect}
Note that the event $\set{W_t|H^+\le 1 \mbox{ and } H^-\le 1, W_1 =
a}$ corresponds to processes that either cross above the barrier
$a/2+1/2$ then below $a/2-1/2$ or vice versa, i.e. processes that
satisfy $H^+_2\le 1$ or $H^-_2\le 1$.  Therefore, from the
inclusion-exclusion principle we have:
\begin{eqnarray*}
\Pr[H^+\le 1 \mbox{ and } H^-\le 1~|~W_1 = a] &=& \Pr[H^+_2\le 1~|~W_1 = a]+
\Pr[H^-_2\le 1~|~W_1 = a] \\ & & -
\Pr[H^+_2\le 1 \mbox{ and } H^-_2\le 1~|~W_1 = a].
\end{eqnarray*}

In order to calculate $\Pr[H^+_2\le 1~|~W_1 = a]$, we apply the
reflection principle twice.  First, we reflect
$W_t$ about the line $a/2+1/2$ when it first hits $a/2+1/2$.  
Call this reflected process $W'_t$.  A process $W_t$
that hits $a/2+1/2$, then hits $a/2-1/2$, then satisfies $W_1 = a$
will correspond exactly to a reflected 
process $W'_t$ that first hits $a/2+1/2$ then 
hits $a/2+3/2$ then achieves $W'_1 = 1$.  

Next, we reflect the process $W'_t$ the first time it hits $a/2+3/2$
about the line $a/2+3/2$; call this new process $W''_t$.  It is easy
to verify that a 
process 
$W_t$ (prior to reflection) that hits $a/2+1/2$, then $a/2-1/2$, then
achieves $W_1 = a$ will correspond exactly to a reflected process $W''_t$ that
satisfies $W''_1 = 2+a$.  Therefore,
\begin{eqnarray*}
\Pr[H^+_2\le 1~|~W_1 = a] &=& \frac{\Pr[H^+_2\le 1\mbox{ and }W_1\in da]}
{\Pr[W_1\in da]}\\ &=& 
\frac{\Pr[W''_1 \in dx(2a+1)]}{\Pr[W_1 \in da]}\\ 
&=&\frac{\phi(2+a)}{\phi(a)}.
\end{eqnarray*}

\begin{comment}
In order to calculate $\Pr[H^+_2\le 1~|~W_1 = a]$, we apply the
reflection principle twice.  After the process $W_t$
$W_t$ first hits $a/2+1/2$ and then hits $a/2-1/2$, 
we reflect the process $W_t$
about the line $a/2-1/2$.  
Call this reflected process $W'_t$.  A process $W_t$
that hits $a/2+1/2$, then hits $a/2-1/2$, then achieves $W_1 = a$
will be reflected to a process $W'_t$ that first hits $a/2+1/2$ and then 
achieves $W'_1 = -1$.  
Next, we reflect the process $W'_t$ the first time it hits $a/2+1/2$
about the line $a/2+1/2$; call this new process $W''_t$.  It is easy
to verify that a 
process $W_t$ (prior to reflection) that hits $a/2+1/2$, then $a/2-1/2$, then
achieves $W_1 = a$ will correspond exactly to a twice reflected 
process $W''_t$ that
achieves $W''_1 = 2+a$.  
Therefore,
\begin{eqnarray*}
\Pr[H^+_2\le 1~|~W_1 = a] &=& \frac{\Pr[H^+_2\le 1\mbox{ and }W_1\in da]}
{\Pr[W_1\in da]}\\ &=& 
\frac{\Pr[W''_1 \in dx(2a+1)]}{\Pr[W_1 \in da]}\\ 
&=&\frac{\phi(2+a)}{\phi(a)}.
\end{eqnarray*}
\end{comment}

Therefore,
\begin{equation*}
\int_{-1}^1\Pr[H^+_2\le 1~|~W_1 = a]\phi(a)da = \int_{-1}^{1}\phi(2+a)da
\le.157305.
\end{equation*}
Similarly, one can show
\begin{equation*}
\int_{-1}^1\Pr[H^-_2\le 1~|~W_1 = a]\phi(a)da = \int_{-1}^{1}\phi(2+a)da \le
.157305.
\end{equation*}

Note that the event
$[H^+_2\le 1 \mbox{ and } H^-_2\le 1~|~W_1 = a]$ corresponds to 
the event $[H^+_3\le 1 \mbox{ or } H^+_3\le 1]$.
From the calcuations in Section~\ref{sec:3-barrier-case-3}, the following bound
can be easily derived:
\begin{equation*}
\int_{-1}^1 \Pr[H^+_3\le 1 \mbox{ or } H^-_3\le 1~|~W_1 = a]\phi(a)da
\ge .01503.   
\end{equation*}

Combining these calculations, we arrive at:
\begin{equation*}
\int_{-1}^1 \Pr[H^-_1\le 1 \mbox { or } H^+_1\le 1~|~W_1 = a]\phi(a)da
\ge .668302
\end{equation*}
\subsection{Totals}
Combining the results of the three cases, we obtain:
\begin{equation*}\Pr[H^+\le 1 \mbox{ or } H^-\le 1]\ge 
.158655\cdot 2+.668302 = .985612.
\end{equation*}  

\subsection{Probability of Three or More Sign Changes}\label{sec:three-changes}
In this section, we prove the following Lemma:
\begin{lemma}\label{lem:three-changes}
$\Pr[H_3^+\le 1 \mbox{ or }H_3^-\le 1]\le .0178.$
\end{lemma}
Since the barriers in $H^+_3$ and $H^-_3$
depend on the value of $W_1$, 
as in the previous section, it will be necessary to
decompose the total probability into probabilities
conditioned on $a=W_1$:
\begin{equation*}
\Pr[H_3^+\le 1 \mbox{ or } H_3^-\le 1] = 
\int_{-\infty}^{\infty} \Pr[H_3^+\le 1 \mbox { or } 
H_3^-\le 1~|~W_1 = a]\phi(a)da.
\end{equation*}
We partition the domain of the integral into three
cases, and calculate the probabilities in each case using
the reflection principle.  

\begin{comment}
A Brownian motion, $W$, begins at 0 and after $t$ time steps 
achieves the value $W_1 = a$.  
Then let $H^+$ be the minimum time that $W$ finishes
reaching the thresholds $a/2+1/2, ~a/2-1/2,~ a/2+1/2$ in that order.
(Let ${H^-}$ be the min time that $W$ reaches the thresholds $a/2-1/2,
~a/2+1/2, ~a/2-1/2$ in that order.)  We define another process $B_t$,
which is a reflection of the process $W_t$ over certain thresholds
(depending on the case).  There are three cases.
\end{comment}

\subsection{Case (i): $a > 1$}
\begin{figure}[h!]
\begin{center}
\epsfig{file=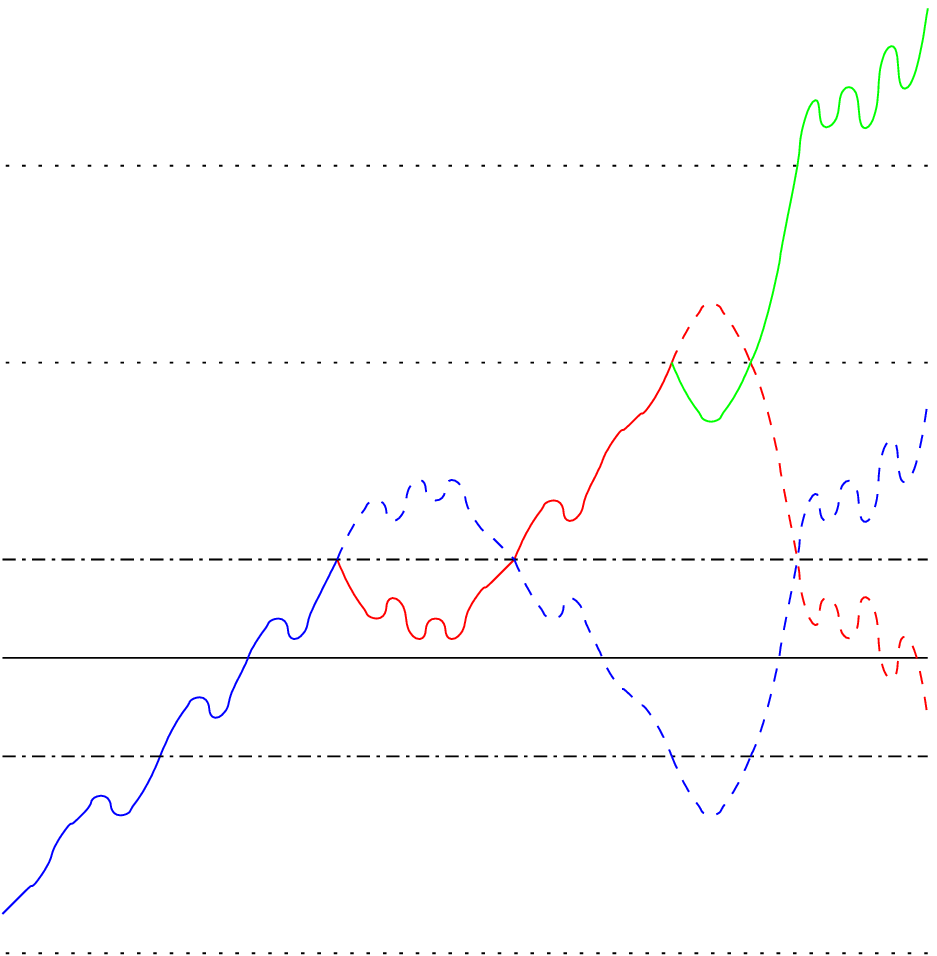, width=6cm}
\caption{\label{interval} Case (i).}
\end{center}
\end{figure}
\noindent

In this case, we only need to calculate the probability that
$H^-_3\le 1$ occurs, since if $H^+_3$ occurs, then $H^-_3$ must also occur.  

To obtain $W'_t$, the process $W_t$ is
reflected the first time it hits $a/2+1/2$, then the first time this
reflected process hits
$a/2+3/2$.  Using reasoning similar to Section~\ref{sec:twice-reflect},
it can be shown that if the process $W_t$ (prior to reflection)
hits $a/2+1/2$, then $a/2-1/2$, then $a/2+1/2$, then satisfies $W_1 = a$
(i.e., it satisfies $H^+_3\le 1$ for fixed $a$), then
it will correspond exactly to a process $W'_t$ that achieves $W'_1 = a+2$.
Therefore,
\begin{eqnarray}
\Pr[H_3^- \le  1~|~ W_1 = a] & = & \frac{\Pr[H^-_3 \le  1, W_1 \in da]}{\Pr[W_1\in da]}\nonumber \\
& = & \Pr[W'_1 \in dx(a+2)]/\Pr[W_1 \in da] \nonumber  \\
& = & \phi(a+2)/\phi(a).\nonumber
\end{eqnarray}
Thus, we have:
\begin{eqnarray}
\int^{\infty}_1 \Pr[H^-_3 \le  1~|~W_1 = a] \phi(a) da. & = & 
\int^{\infty}_1 \phi(a+2) da \nonumber \\ 
& \leq & .0013499. \nonumber  
\end{eqnarray}

\subsection{Case (ii):  $a\le -1$}
Analogous to Case (i).  

\subsection{Case (iii): $-1 < a < 1$}\label{sec:3-barrier-case-3}
By the inclusion-exclusion principle, we have:
\begin{eqnarray*}
\Pr[H^+_3\le 1\mbox{ or } H^-_3\le 1~|~W_1 = a ] &=& 
\Pr[H^+_3\le 1~|~W_1 = a]+\Pr[H^-_3\le 1~|~W_1 = a]\\ 
& &-
\Pr[H^+_3\le 1 \mbox{ and }H^-_3\le 1~|~W_1 = a].\\ 
\end{eqnarray*}
\begin{figure}[h]
\begin{center}
\epsfig{file=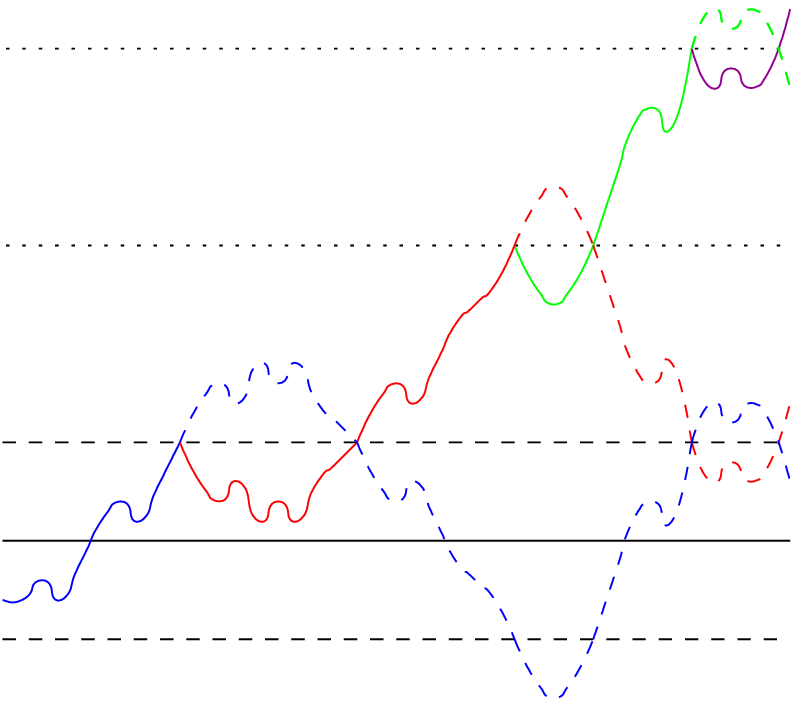, width=6cm}
\caption{\label{case_iib} Case (ii), $H^+$.}
\end{center}
\end{figure}
\noindent
First, we calculate $\Pr[H^+_3\le 1~|~W_1 = a]$.  
The process $W'_t$ is obtained by
reflecting $W_t$ the first time it hits $a/2+1/2$, then
the first the reflected process hits
$a/2+3/2$, then the first time the twice reflected process hits
$a/2+5/2$.  If $W_t$ (prior to reflection) 
hits $a/2+1/2$, then $a/2-1/2$, then $a/2+1/2$, then achieves $W_1 = a$,
then it will correspond exactly to 
a thrice reflected process $W'_t$ that satisfies $W'_1 = 3$.  


We want to calculate:
\begin{eqnarray}\label{eq:case-ii-i}
\int^{1}_{-1} \Pr[H^-_3 \le  1~|~W_1 = a] \phi(a) da. 
\end{eqnarray}
We have: 


\begin{eqnarray}
\Pr[H^-_3 \le  1~|~ W_1 = a] & = & 
\frac{\Pr[H^-_3 \le  1, W_1 \in da]}{\Pr[W_1\in da]}\nonumber \\
& = & \Pr[W'_1 \in dx(3)]/\Pr[W_1\in da] \nonumber  \\
&=& \phi(3)/\phi(a).\nonumber
\end{eqnarray}
Thus, we have:
\begin{eqnarray}
\int^{1}_{-1} \Pr[H^-_3 \le 1~|~W_1 = a] \phi(a) da. & = & 
 2\int^{1}_0 \phi(3) da \nonumber \\ 
& \leq & .0088637. \nonumber  
\end{eqnarray}  

\begin{figure}[h!]
\begin{center}
\epsfig{file=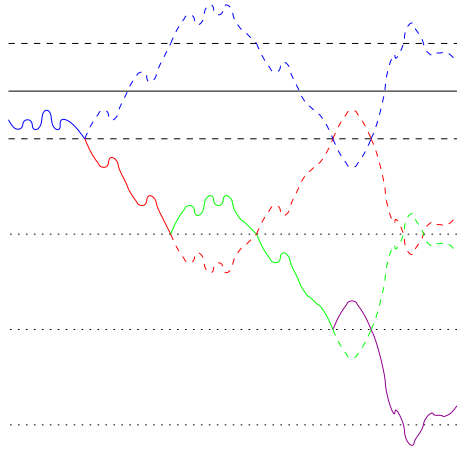, width=6cm}
\caption{\label{case_iia} Case (ii), $H^-$.}
\end{center}
\end{figure}
Now we calculate $\Pr[H^-_3\le 1~|~W_1 = a]$.  
In this case, the process $W'_t$ is obtained by
reflecting $W_t$ the first time it hits $a/2-1/2$, then the first
time the reflected process hits  $a/2-3/2$, then the first time
the twice reflected process hits $a/2-5/2$.  A process $W_t$
that hits $a/2-1/2$, then $a/2+1/2$, then $a/2-1/2$, then satisfies $W_1 = a$
(i.e., it satisfies $H^-_3\le 1$ for fixed $a$) 
will correspond exactly to a 
thrice reflected process $W'_t$
that satisfies $W'_1 = -3$.  
We want to compute:
\begin{eqnarray}\label{eq:case-ii-ii}
\int^{1}_{-1} \Pr[H^-_3 \le 1~|~W_1 = a] \phi(a) da. 
\end{eqnarray}
We have: 
\begin{eqnarray}
\Pr[H^-_3 \le  1~|~ W_1 = a] & = & \Pr[H^-_3 \le  1, W_1 \in da]/
\Pr[W_1\in da]\nonumber \\
& = & \Pr[W'_1\in dx(-3)]/\Pr[W_1\in da] \nonumber  \\
&=& \phi(-3)/\phi(a).\nonumber
\end{eqnarray}
Thus, we have:
\begin{eqnarray}
\int^{1}_{-1} \Pr[H^-_3 \le  t~|~W_t = a] \phi(a) da. 
& = & 2\int^{1}_{0} \phi(-3) da \nonumber \\ 
& \le & .0088637. \nonumber  
\end{eqnarray}


Thus, a naive bound on the probability of three sign changes
would be to add expressions~(\ref{eq:case-ii-i}) and (\ref{eq:case-ii-ii}): 
\begin{eqnarray*}
\int_{-1}^1\Pr[H^{-}_3\le 1 \mbox{ or } H^{+}_3\le 1~|~W_1 = a]\phi(a)da
&\le& \int_0^1 \Pr[H^{-}_3\le 1|W_1 =a]\phi(a)da\\
& &+\int_0^1 \Pr[H^{+}_3\le 1|W_1 =a]\phi(a)da 
= .0017728.
\end{eqnarray*}
The above bound is an overestimate of the probability, because the
event $[H^-_3\le 1 \mbox{ and } H^+_3\le 1~|~W_1 = a]$ is contained in
both (\ref{eq:case-ii-i}) and (\ref{eq:case-ii-ii}).

We now calculate $\Pr[H^-_3\le 1 \mbox{ and } H^+_3\le 1~|~W_1 = a]$.
Note that the event $\set{W|H^-_3\le 1 \mbox{ and } H^+_3\le 1}$
occurs when there are at least four sign changes; either
$H^+_4\le 1$ or $H^-_4\le 1$ occurs, or possibly both.    
Using the same argument as we did for $H^+_3$, $H^-_3$, it can
be shown that:
\begin{equation*}
\int_{-1}^1 \Pr[H^-_4\le 1~|~W_1=a]\phi(a)da + \int_{-1}^1\Pr[H^+_4\le 1~|~W_1=a]
\phi(a)da = 2\int_{-1}^1 \phi(4+a)da \ge
.00269922
\end{equation*}
and that
\begin{equation*}
\int_{-1}^1 \Pr[H^+_5\le 1~|~W_1=a]\phi(a)da + 
\int_{-1}^1 \Pr[H^-_5\le 1~|~W_1 = a]\phi(a)da = 
4\int_0^1\phi(5)da \le 5.94688 \cdot 10^{-6}. 
\end{equation*}

Again applying the inclusion-exclusion principle, we have:
\begin{eqnarray*}
\int_{-1}^1 \Pr[H^-_3\le 1 \mbox{ and } H^+_3\le 1~|~W_1 = a] \phi(a)da&=& 
\int_{-1}^1 \left(\Pr[H^-_4\le 1~|~W_1=a]+\Pr[H^+_4\le 1~|~W_1=a]\right)\phi(a)da\\ & &-
\int_{-1}^1\Pr[H^-_4\le 1 \mbox{ and } H^+_4\le 1~|~W_1 = a]\phi(a)da\\
&\le& 
\int_{-1}^1\left(\Pr[H^-_4\le 1~|~W_1=a]+\Pr[H^+_4\le 1~|~W_1=a]\right)\phi(a)da\\ & &-
\int_{-1}^1\left(\Pr[H^-_5\le 1~|~W_1 = a]+\Pr[H^+_5\le 1~|~W_1 = a]\right)
\phi(a)da\\
&\le&
.0026328.
\end{eqnarray*}
Therefore:
\begin{equation*}
\int_{-1}^1\Pr[H^-\le 1 \mbox{ and } H^+\le 1]\phi(a)d_a\le .0176734-.00263828\le .015035.
\end{equation*}

\subsection{Totals}

Combining the results of the three cases, we arrive at: 
\begin{equation*}
\Pr[H^+_3\le 1 ~{\text{or}} ~H^-_3\le 1]\le
.015035 + .0013499\cdot 2 = .017735.
\end{equation*}

\section{From Brownian Motion to Discrete Random Walks}\label{sec:discrete}

The randomized rounding procedure for our algorithm involves a
discrete random walk; we have proven Lemmas~\ref{lem:one-change}
and~\ref{lem:three-changes} for the continuous process, Brownian
motion.  We show in this section that the discretized random walk of
the rounding procedure will also satisfy Lemmas~\ref{lem:one-change}
and~\ref{lem:three-changes}.

Suppose $W_t$ is a Brownian motion.  As we showed earlier, the
discretized random walk of $s$ steps, 
$w_1, \ldots, w_s$, can be modeled as
the sequence: $w_1 = W_{1}, w_2 = W_{2/s}, w_3 = W_{3/s},\ldots, w_s =
W_1$.

First, consider the question of whether Lemma~\ref{lem:three-changes}
implies that $w_1,\ldots, w_s$ also does not touch the sequence of
barriers $\frac{w_s}{2}+\frac{1}{2}, \frac{w_s}{2}-\frac{1}{2},
\frac{w_s}{2}+\frac{1}{2}$ before time $t = 1$.  Certainly, if $W_t$
does not hit this sequence of barriers, then its discretized version
also does not hit these three barriers, since $w_s = W_1$.  Therefore
Lemma~\ref{lem:three-changes} holds for the discrete random walk as well.

Now consider the question of whether Lemma~\ref{lem:one-change}
implies that $w_1, \ldots, w_s$ hits either of the barriers
$\frac{w_s}{2}+\frac{1}{2}$ or $\frac{w_s}{2}-\frac{1}{2}$.  Note that
if $W_t$ hits the barrier $b$ at time $\tau_b$, it is not necessarily
true that there exists an $i$ such that $w_i \ge b$, since $W_t$ could
have hit $b$ at some time between the steps of the discretized walk.
Therefore, Lemma~\ref{lem:one-change} cannot be immediately adapted to
proving properties of the discretized walk.  We now prove that the
Lemma is true for the discretized walk when the number of steps is a
sufficiently large constant.

Recall that random variables $H^+$ and $H^-$ were defined as the first
times that the Brownian motion $W_t$ hits the barrier defined by
$\frac{W_1}{2}+\frac{1}{2}$ and $\frac{W_1}{2}-\frac{1}{2}$,
respectively.  We slightly strengthen these conditions by defining
random variables $\tilde{H}^+$ and $\tilde{H}^-$ to be the first
times that $W_t$ hits the barriers $\frac{W_1}{2}+\frac{1}{2}+\eta$
and $\frac{W_1}{2}-\frac{1}{2}-\eta$, respectively, for some very
small constant $\eta$.

Since $\eta$ will be chosen to be very 
small, it will not have a large impact
on the distributions of $\tilde H^+$ and $\tilde H^-$
relative to $H^+$ and $H^-$.  
The proof of the following lemma involves the same
calculations as in the proof of
Lemma~\ref{lem:one-change}.  
\begin{lemma}\label{lem:one-change-eta}
For $\eta>0$ chosen sufficiently small,
\begin{equation*}
\Pr[\tilde{H}^+\le 1~{\text{or}}~\tilde{H}^-\le 1] \ge 0.9855.
\end{equation*}
\end{lemma}

We use the above Lemma to
prove that if the continuous process $W_t$ hits
the barrier $\frac{a}{2}+\frac{1}{2}+\eta$, then 
the discrete random walk will hit the
barrier $\frac{a}{2}+\frac{1}{2}$ with high probability.  
The case for the barrier $\frac{a}{2}-\frac{1}{2}$ is similar.  
\begin{lemma}\label{lem:discretize-2}
If the number
of steps $s$ of the discretized random walk
satisfies $s\ge \frac{c}{\eta^2}$ for some constant $c$, 
then:
\begin{equation*}
\Pr\left[w_{\lceil s \cdot \tau_{\frac{a}{2}+\frac{1}{2}+\eta} \rceil}\ge
  \frac{a}{2}+\frac{1}{2} ~|~ \tau_{\frac{a}{2}+\frac{1}{2}+\eta}\le 1 \right]
\ge .997,
\end{equation*}
where $W_1 = a$ and $\tau_{\frac{a}{2}+\frac{1}{2}}$ is the time the 
continuous process hits the barrier $\frac{a}{2}+\frac{1}{2}$.  
\end{lemma}
\begin{proof}
Let $b= \frac{a}{2}+\frac{1}{2}$ be the barrier of interest.
Since $b$ depends on
value of $W_1 = a$, as in Sections~\ref{sec:one-change} and~
\ref{sec:three-changes}, we will work with probabilities
conditioned on the event $\set{W|W_1 = a, \tau_b\le 1}$.  

Note that
{\bf (a)} $a = W_1\sim N(0, 1)$; therefore, 
with probability at least $.999$, $|a|\le 10$
and $b +\eta< 6$.   
Also, {\bf (b)} 
the probability that 
$\tau_{b+\eta}\le 1-c$, for some constant $c>0$, conditioned
on $\tau_{b+\eta}\le 1$, is at least 0.999.  This is because
the density function of $\tau_{b+\eta}$ conditioned
on $W_1 = a$ 
is given by:
\begin{equation*}
\Pr\left[\tau_{b+\eta}\in dt~|~W_1 = a\right] = 
\frac{b+\eta}{\sqrt{2\pi}t^{3/2}}\exp
\left(-\frac{(b+\eta)^2}{2t}\right)\frac{1}{\sqrt{1-t}}
\phi\left(\frac{b+\eta-a}{\sqrt{1-t}}\right).
\end{equation*}
Therefore,
\begin{eqnarray}
\Pr\left[\tau_{b+\eta}\in [1-c, 1]~|~W_1=a, |a|\le 10,
\tau_{b+\eta}\le 1\right] 
&\le& \frac{1}{(.999)^2}
\int_{1-c}^1 
\frac{b+\eta}{\sqrt{2\pi}t^{3/2}}\exp
\left(-\frac{(b+\eta)^2}{2t}\right)\\
& &\cdot \frac{1}{\sqrt{1-t}}
\phi\left(\frac{b+\eta-a}{\sqrt{1-t}}\right)dt\\
&\le& \frac{6}{((.999)^2\cdot \pi)}\int_{1-c}^1
\frac{1}{\sqrt{1-t}}dt \\
&\le& \frac{6}{((.999)^2\cdot \pi)} \left(|^{1}_{1-c} ~(-2 \sqrt{1-t})\right)\\
&\le& \frac{6}{((.999)^2\cdot \pi)} (2 \sqrt{c}) \le .001,
\end{eqnarray}
for appropriately chosen $c$.  In particular $c \approx 10^{-9}$ is
sufficiently small.

If $\tau_{b+\eta}$ is the time that the process $W_t$
hits the barrier $b+\eta$, let the index 
$\lceil s\cdot\tau_{b+\eta}\rceil$ denote the step in the discretized
random walk that immediately follows $\tau_{b+\eta}$.
The value of this step is
$w_{\lceil s\cdot \tau_{b+\eta}\rceil} = W_{\lceil s\cdot \tau_{b+\eta}\rceil/s}$.
Intuitively, this value should be very close to $b+\eta$ if the number
of steps is sufficiently large.  
Indeed, we will prove the lemma by showing that
if the number of steps in the discrete random
walk satisfies $s\ge \frac{20}{c\eta^2}$, then
\begin{equation*}
\Pr\left[W_{\frac{\lceil s\cdot\tau_{b+\eta}\rceil}{s}}\ge \frac{a}{2}+\frac{1}{2} ~|~
\tau_{b+\eta}\le 1-c, W_1 = a, |a|\le 10\right]
\ge .999.
\end{equation*}

Suppose that $W$ is conditioned on reaching the barrier $b+\eta$ at
time $T$ and that $W$ is restricted to satisfying $W_1 = a$.  We use
basic properties of the distribution of the increments of a Brownian
Bridge (see \cite{jChang} for details) to show that
the value of a Brownian motion at time $T+t<1$, under the condition
that $W_T=b+\eta$ and $W_1 = a$, has the following distribution:
\begin{equation}\label{eq:bridge-distribution}
W_{T+t}|W_T = b+\eta, W_1 = a \sim N
\left(b+\eta-\frac{t (b+\eta-a)}{1-T}, 
\frac{t}{1-T}\cdot (1-T-t)\right).
\end{equation}  
Note that $\lceil s\cdot\tau_{b+\eta}\rceil$ is
the index of the closest step in the discretization to $\tau_{b+\eta}$
and that $\lceil s\cdot\tau_{b+\eta}\rceil/s - \tau_{b+\eta}\le 1/s
\le c\eta^2/20$.  
If $a\le 10, T< (1-c)$ and $s\ge 20/(c\eta^2)$, 
then Equation (\ref{eq:bridge-distribution})
implies that 
$w_{\lceil s\cdot \tau_{b+\eta}\rceil}=W_{\lceil s\cdot\tau_{b+\eta}\rceil/s}$ 
is distributed with mean 
at least $b+\eta/2$
and variance at most $\eta^2/20$.
Thus, if $s\ge 20/(c\eta^2)$,
\begin{equation*}
\Pr[w_{\lceil s\cdot\tau_{b+\eta}\rceil}\le b~|~
|a|\le 10,\tau_{b+\eta}\le 1-c, W_1 = a] \leq .001.
\end{equation*}
The Lemma follows.
\end{proof}

Lemma~\ref{lem:discretize-2} can thus be applied to prove Theorem
\ref{sign_change_thm}.

\section{Correlated walks}\label{sec:correlated}

To prove an approximation ratio of our rounding algorithm, we need to
show that the positions of $x_i$ and $x_j$ (corresponding to the
constraint $x_j -x_i \equiv d_{ij} (\bmod ~s)$), determined by the
random walks $w^i$ and $w^j$, are close to the required distance if
the vectors $v_i^0$ and $v_j^{d_{ij}})$ are close.  In other words,
  without loss of generality, let us assume that for a fixed
  constraint, we have $d_{ij} = 0$.  Then our goal is to show that the
  distance between the two positions assigned by our rounding
  procedure to $x_i$ and $x_j$ are close if the vectors $v_i^0$ and
  $v_j^0$ are close.  After extensive computational investigation (on
  solutions obeying the constraints of $(P^+)$), we believe the
  following conjecture holds.
\begin{conjecture}
  In our rounding scheme, the expected distance between $x_i$ and
  $x_j$ is bounded above by $\frac{\theta}{2\pi}$ if both $w^i$ and
  $w^j$ each have exactly one extreme sign change.
\end{conjecture}
Proving the above conjecture would lead to an approximation guarantee
slightly below $\alpha_{GW} = .87856$, because we do not have an
extreme sign change with probability 1.  

We can show that if $v_i^0$ and $v_j^0$ have a small angle, then the
two walks are (globally) close to each other in the sense that the area between
the two walks is small.  However, this does not immediately lead to a
proof that the positions of their extreme sign changes are close.  
\begin{lemma}
Given two unit vectors $x$ and $y$ with angle $\theta$, and a vector
$r \in \mathbbm{R}^n$ with each coordinate drawn from $\N(0,1)$, then,
\vspace{-2mm}
\begin{eqnarray*}
\text{E}[|x \cdot r - y \cdot r|] = \frac{2 \sqrt{2}}{\sqrt{\pi}}\sin{\frac{\theta}{2}}.
\end{eqnarray*}
\end{lemma}
\begin{proof}
Let $x = (\cos{\frac{\theta}{2}}, ~\sin{\frac{\theta}{2}})$ and $y =
(\cos{\frac{\theta}{2}}, - \sin{\frac{\theta}{2}})$.  Let $r = (r_1, r_2)$. 
\begin{eqnarray}
\text{E}[|x \cdot r - y \cdot r|] = |2 r_2 \sin{\frac{\theta}{2}}| =
\text{E}[|r_2|]\cdot 2 \sin{\frac{\theta}{2}}. \nonumber
\end{eqnarray}
The expected value of $r_2$ given that it is non-negative is
$\frac{\sqrt{2}}{\sqrt{\pi}}$.  Since $\sin{\frac{\theta}{2}}$ is
always non-negative for $\theta$ from 0 to $\pi$, the above statement follows by linearity of expectation.
\end{proof}

\vspace{3mm}

If we consider the random walks on the interval $[0,1]$ (i.e. we map
the interval $[0,2]$ to the smaller interval $[0,1]$), then the
expected area between the two walks is
$\frac{2\sqrt{2}}{\sqrt{\pi}}\sin{\frac{\theta}{2}}$.  Thus, as the
contribution to the objective function increases, the two walks
converge and the positions assigned to them by the rounding procedure
should converge to one another.

\begin{figure}[h!]
\begin{center}
\epsfig{file=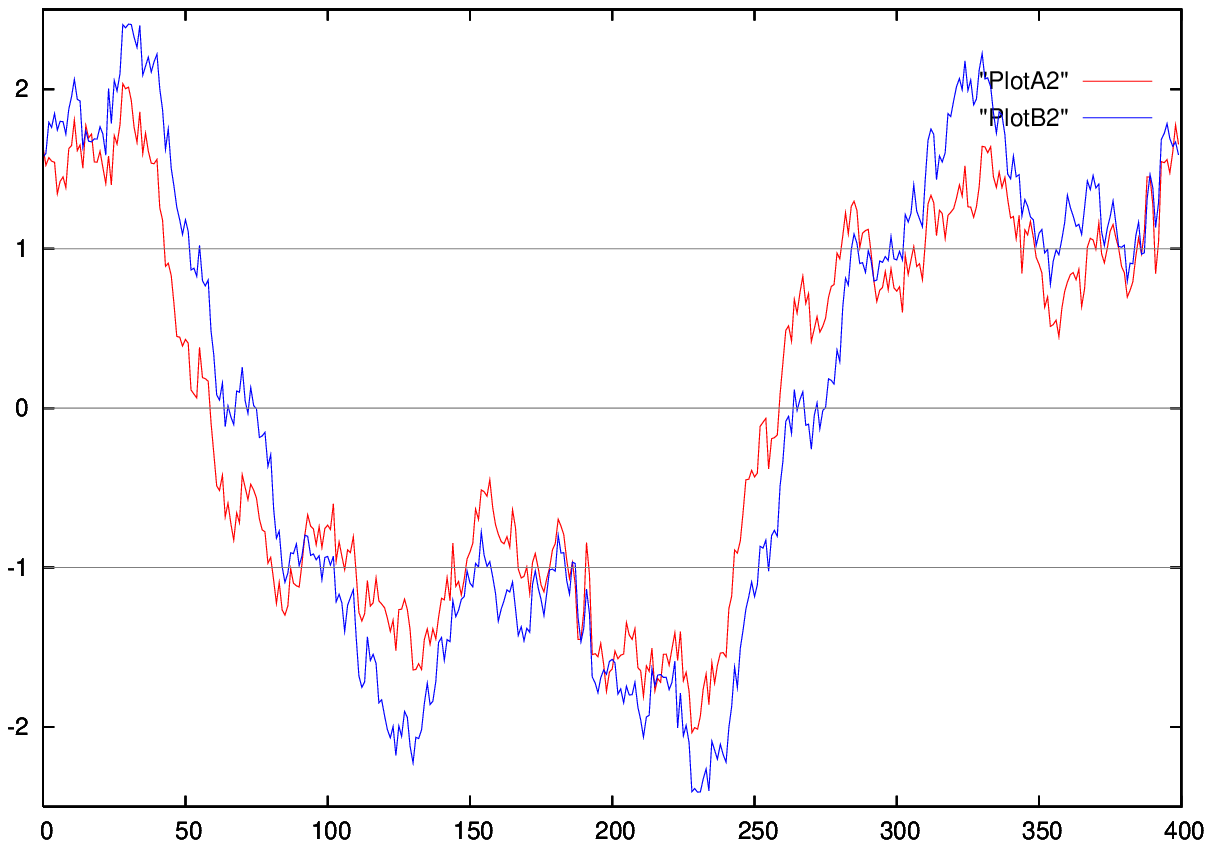, width=6cm}
\epsfig{file=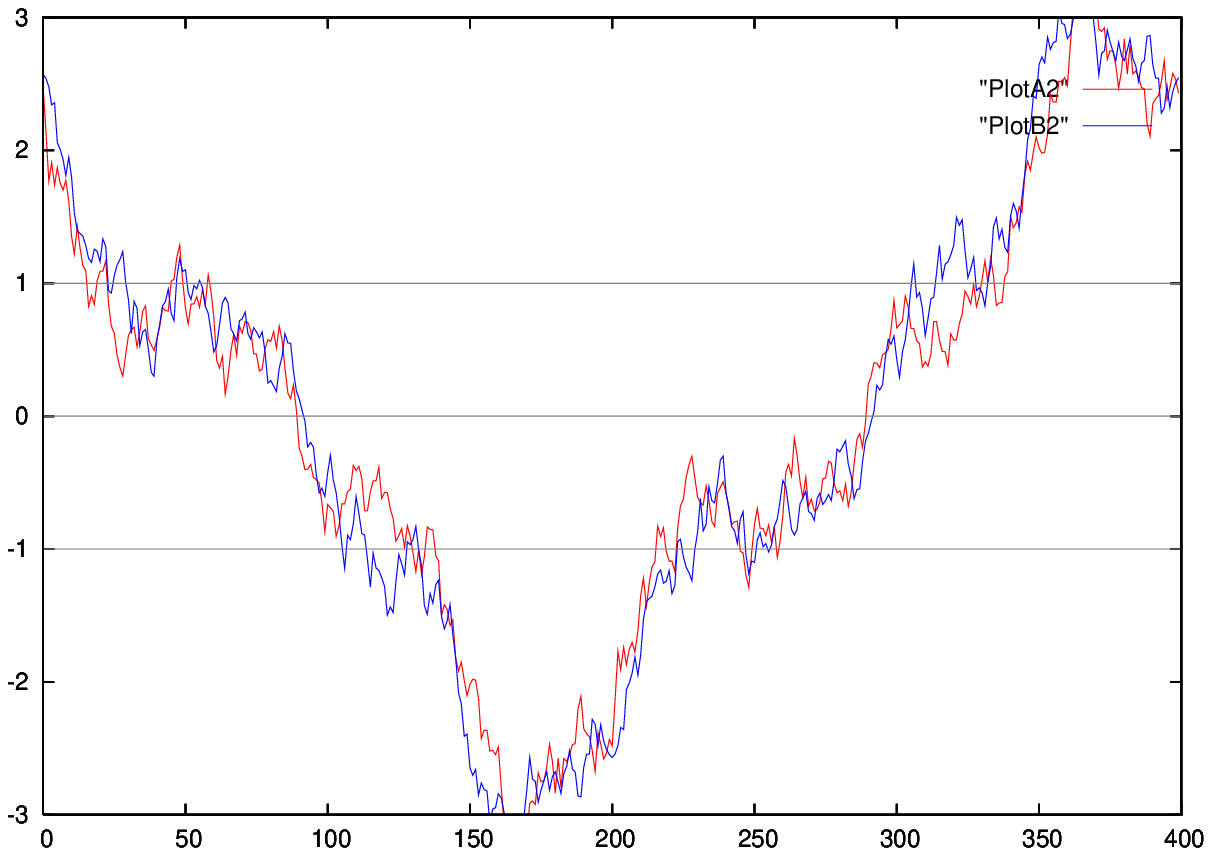, width=6cm}
\caption{In the first example,
      $\cos{\theta} = .86$.  In the second, $\cos{\theta} = .945455$.}
\end{center}
\end{figure}

 \begin{figure}[h]
\begin{center}
  \epsfig{file=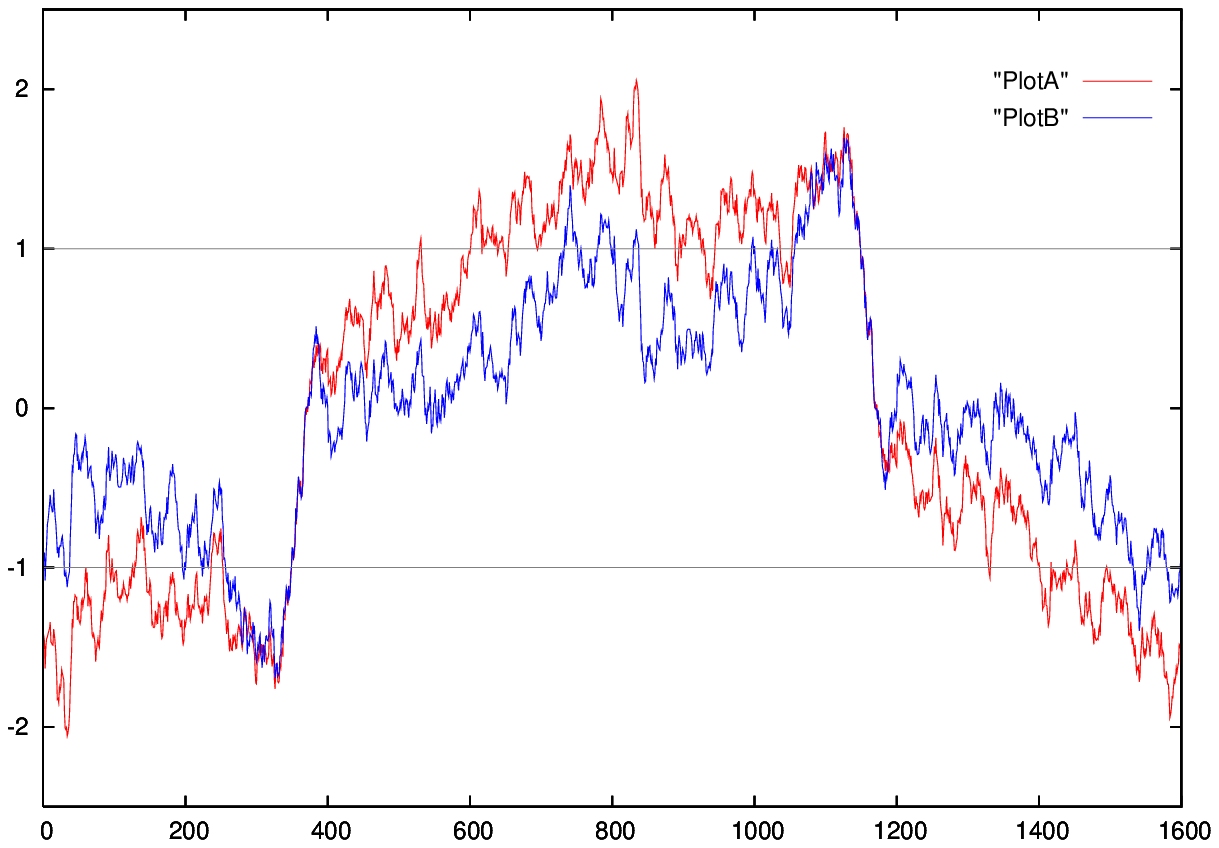, width=7.4cm} 
\epsfig{file=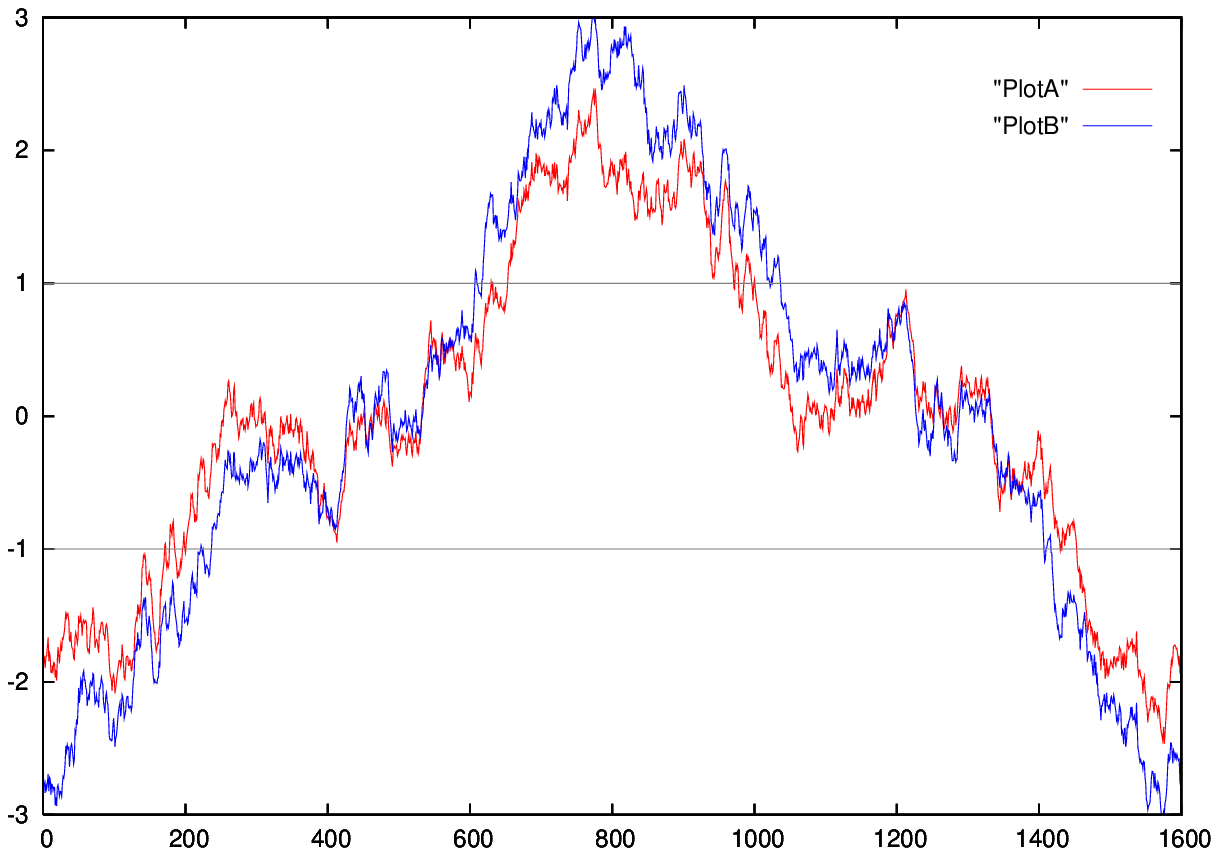, width=7.4cm} 
\epsfig{file=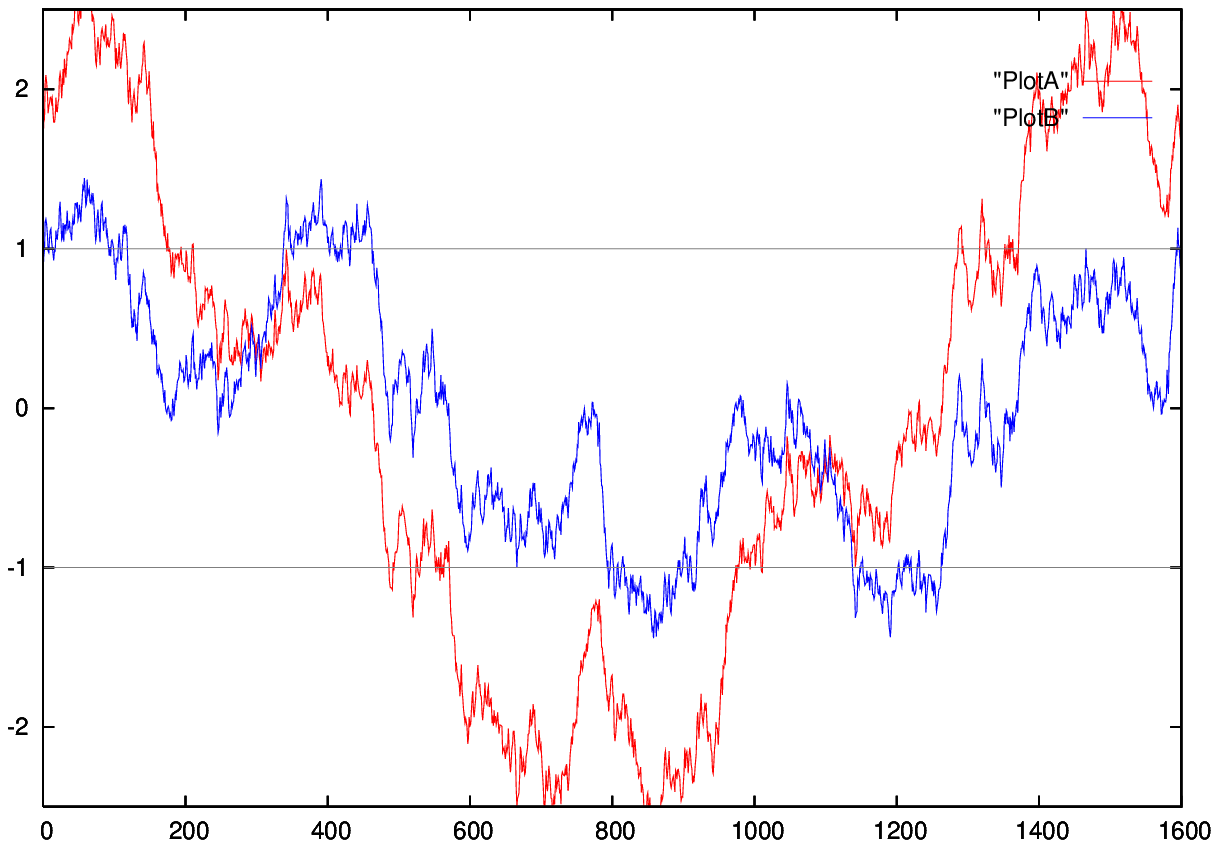, width=7.4cm}
\caption{More examples of correlated walks.}
\end{center}
\end{figure}

\clearpage

\section*{Acknowledgements} We would like to thank Martin Becker and
Larry Shepp for helpful discussions about Brownian motion.  Most of
this work was done in 2007 at the Max-Planck-Institut f\"ur Informatik
in Saarbr\"ucken, Germany.  \bibliography{modp}

\end{document}